\documentclass[conference]{IEEEtran}
\usepackage{amsmath,amssymb,amsthm}
\usepackage{graphicx}
\usepackage{graphics}
\usepackage{epstopdf}
\usepackage{mathtools}

\newtheorem{remark}{Remark}

\newtheorem{theo}{Theorem}

\newtheorem{corollary}{Corollary}

\ifCLASSINFOpdf
\else
\fi

\begin{document}

\title{A Multi-Service Oriented Multiple-Access Scheme for Next-Generation
Mobile Networks}

\author{\IEEEauthorblockN{Nassar Ksairi,
Stefano Tomasin and M\'erouane Debbah}
Mathematical and Algorithmic Sciences Lab,\\
France Research Center, Huawei Technologies Co. Ltd.,
Boulogne-Billancourt, France.
\IEEEauthorblockA{Emails:
\{nassar.ksairi, stefano.tomasin, merouane.debbah\}@huawei.com}}

\maketitle

\begin{abstract}
One of the key requirements for fifth-generation (5G) cellular networks is their
ability to handle densely connected devices with different quality of service
(QoS) requirements. In this article, we present multi-service oriented multiple
access (MOMA), an integrated access scheme for massive connections with
diverse QoS profiles and/or traffic patterns originating from
both handheld devices and machine-to-machine (M2M) transmissions.
MOMA is based on a) establishing separate classes of users based on relevant
criteria that go beyond the simple handheld/M2M split,
b) class dependent hierarchical spreading of the data signal and
c) a mix of multiuser and single-user detection schemes at the receiver.
Practical implementations of the MOMA principle are provided for base stations
(BSs) that are equipped with a large number of antenna elements.
Finally, it is shown that such a massive-multiple-input-multiple-output (MIMO)
scenario enables the achievement of all the benefits of MOMA even with a
simple receiver structure that allows to concentrate the receiver complexity
where effectively needed.
\end{abstract}

\IEEEpeerreviewmaketitle

\section{Introduction}
\label{sec:intro}

In the aim of deploying the Internet of Things (IoT), designing a unified radio
access technique for both machine-to-machine (M2M) communications and
handheld mobile devices is a challenging problem. One major issue is the
difference in terms of traffic patterns and QoS requirements~\cite{smart_city}
between these two types of communications. Another issue is the large number of
IoT devices required to be simultaneously served. Further difficulties arise
from the fact that M2M transmissions do not all have the same QoS profile and
traffic characteristics~\cite{rel_13,ngmn}. 

To address some of these challenges, the Third Generation Partnership Project
(3GPP) has started to add M2M-type communications support into the radio access
subsystem of LTE. Several proposals emerged within this work. Two of them are
dubbed LTE for Machine-Type Communications (LTE-M) and Narrow-band LTE-M
(NB LTE-M), each introducing a new user equipment (UE) category, the so-called
Cat.~1.4MHz for LTE-M and Cat.~200kHz for NB LTE-M~\cite{lte_m}.
As their respective names indicate, these new UE categories restrict M2M
transmissions to a small subband of the available bandwidth that is orthogonal
to the broadband users. The same principle is used in
Filtered-OFDM~\cite{f_ofdm} with the difference that in Filtered-OFDM,
subband-based filtering is applied to enable the use of a different
transmission time interval (TTI) and OFDM numerology on the M2M
subband. Other solutions consist in providing a
separate network for M2M connections. Examples include LoRa$^{\textrm{TM}}$ and
SIGFOX$^{\textrm{TM}}$~\cite{smart_city}, which both operate in the unlicensed
frequency bands. While the physical layer of
SIGFOX$^{\textrm{TM}}$ is based on frequency-division multiple access (FDMA)
with ultra narrow band sub-channels, the technology adopted in
LoRa$^{\textrm{TM}}$~\cite{lora_phy} employs is a mix of FDMA and
of chirp spread spectrum~\cite{coverage_enhancement}. 
However, none of the existing solutions is able to meet the following crucial
requirements all at once.\\
{\bf 1. Denser IoT deployment:}
The existing proposals offer significant improvement over current
cellular standards in terms of support for IoT access but there is a
need to support much larger numbers of simultaneous M2M transmissions.
This goal should be met without sacrificing the QoS of broadband mobile
services.\\
{\bf 2. Multi-class users/services:}
Not enough attention has been paid to the different QoS and traffic profiles
within the class of IoT devices. Indeed, many of these devices
\emph{will not all be battery limited sensors and will not only emit small
packets of data}~\cite{rel_13}.
Moreover, there is a need to distinguish from within the services running on
handheld devices those that have traffic characteristics and
data rate requirements, e.g. short messages, previously considered to be typical
of IoT services. Significant gains in resource utilization efficiency are
expected from a multiple-access scheme that treats devices with different QoS
profiles, whether handheld devices or IoT machines, as belonging
to separate classes of users.\\
{\bf 3. Flexibility in resource assignment:}
The new multiple-access scheme should be flexible in assigning resources
to the different classes and to the different users within each class.\\
{\bf 4. Efficiency in resource utilization:}
The new multiple-access scheme should be more efficient in resource utilization
than the orthogonal schemes adopted in all the existing proposals.

\section{Multi-service Oriented Multiple Access}
\label{sec:description}
\begin{figure*}
 \centering
 \includegraphics[width=0.9\hsize]{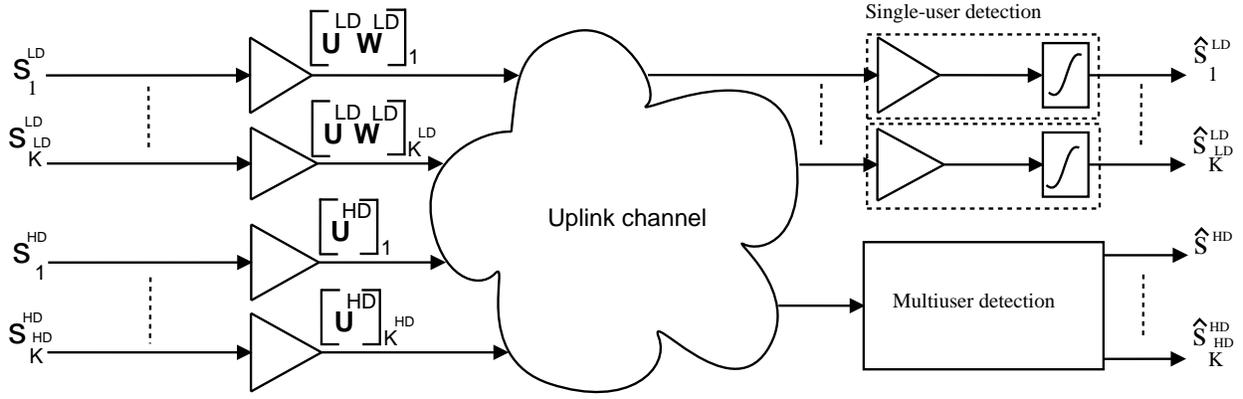}
 \caption{MOMA transmitters of two classes of users ($L=2$).}
 \label{fig:moma_principle}
\end{figure*}
Consider a cell containing a base station (BS) required to serve $K$ users in
the \emph{uplink}\footnote{The MOMA principle can be straightforwardly
extended to the downlink.}. Assume that these users are grouped into
$L\geq2$ classes defined as follows.
\begin{itemize}
\item {\bf One maximum data rate (HD) class:} Here, HD stands for ``high data
rate''. For this class the objective is to obtain a data rate \emph{as high as
possible} for a given number of users. Typically, these users are associated to
data hungry applications on handheld devices.
\item {\bf $\boldsymbol{L-1}$ constant data rate (LD) classes:} Here, LD stands
for ``low data rate''. For these classes, the objective is to accommodate
\emph{as many users as possible with granted data rates}
$r_l^{\mathrm{LD}}$ that satisfy
\begin{equation}
 \label{eq:r_ld_order}
r_1^{\mathrm{LD}}>r_2^{\mathrm{LD}}>\ldots>r_{L-1}^{\mathrm{LD}}\:.
\end{equation}
These users are typically associated with fixed low-rate transmissions from both
handheld devices (such as short messages) and from different types of IoT
services.
\end{itemize}
In the sequel, we use $\mathcal{K}^{\mathrm{HD}}\subset\{1,2,\ldots,K\}$ (resp.
$\mathcal{K}_l^{\mathrm{LD}}$) to designate the indexes
of the users of the HD (resp. the $l$-th LD) class. 
We also define $K^{\mathrm{HD}}\stackrel{\mathrm{def}}{=}
\left|\mathcal{K}^{\mathrm{HD}}\right|$,
$K_l^{\mathrm{LD}}\stackrel{\mathrm{def}}{=}
\left|\mathcal{K}_l^{\mathrm{LD}}\right|$,
$\mathcal{K}^{\mathrm{LD}}\stackrel{\mathrm{def}}{=}
\bigcup_{l\in\{1\cdots L-1\}}\mathcal{K}_l^{\mathrm{LD}}$,
$K^{\mathrm{LD}}\stackrel{\mathrm{def}}{=}
\left|\mathcal{K}^{\mathrm{LD}}\right|$
and
$\mathcal{K}\stackrel{\mathrm{def}}{=}
\mathcal{K}^{\mathrm{HD}}\cup\mathcal{K}^{\mathrm{LD}}$.
 
Since what matters for HD users is maximizing throughput,
proper scheduling techniques will limit the number of
simultaneously-transmitting users, thus it is reasonable to assume that
$K^{\rm HD}$ will be small.
This observation leads us to use multiuser detection techniques for this class
of users. On the other hand, LD users will be quite numerous so that multiuser
detection would be too complex. Therefore, we propose the use of single-user
detection for this class of users. Moreover, to allow for high-performance HD
connections, we want to limit interference from LD transmissions on HD received
signals, we thus let HD and LD users to be quasi-orthogonal to each other.
Lastly, this quasi-orthogonality is implemented in the code domain using a novel
{\it class-dependent hierarchical spreading} scheme.

In the rest of this article we show how such a multiple-access scheme has
the desirable property that the available radio resources are used in a
\emph{flexible and efficient} manner to allow connecting not only a large number
of IoT devices, but also a large number of handheld devices running relatively
low data rate services, while guaranteeing the satisfaction of the broadband
services with high QoS requirements. This flexibility and
efficiency in using the available resources is to be contrasted with the
existing access solutions for IoT, such as LoRa$^{\textrm{TM}}$,
SIGFOX$^{\textrm{TM}}$ and LTE-M, where the resources reserved for IoT
services are under-used and their proportion to the overall resources
cannot be dynamically adjusted.

\section{MOMA Transmitter}

The MOMA transmitter ($L=2$) is illustrated in Fig.~\ref{fig:moma_principle}.
\subsection{Class Dependent Hierarchical Spreading}

Let $\mathbf{U}$ be an orthogonal-code matrix, e.g. Walsh-Hadamard, whose
columns are $N$-long spreading codes.
In MOMA for $L$ classes of users, the set of columns of matrix $\mathbf{U}$ is
divided into $L$ disjoint subsets that form $L$ matrices, namely
$\mathbf{U}^{\mathrm{HD}}$ and
$\{\mathbf{U}_l^{\mathrm{LD}}\}_{l\in\{1\cdots L-1\}}$ with respective
dimensions $N\times N^{\mathrm{HD}}$ and
$\{N\times N_l^{\mathrm{LD}}\}_{l\in\{1\cdots L-1\}}$.
The $N^{\mathrm{HD}}$ spreading sequences of
$\mathbf{U}^{\mathrm{HD}}$ will be assigned to $K^{\mathrm{HD}}$ HD users.
Typically, $K^{\mathrm{HD}}\leq N_l^{\mathrm{LD}}$
so that the HD users could be scheduled in a quasi-orthogonal manner. This
inequality could however be violated in the case where the spatial dimension
is exploited, as we will see. The $N_l^{\mathrm{LD}}$ columns
of $\mathbf{U}_l^{\mathrm{LD}}$ will be shared among the users
of the $l$-th LD class. The scenario of interest for MOMA is when
$K_l^{\mathrm{LD}}>N_l^{\mathrm{LD}}$, i.e. LD resources are overloaded, and
when the following inequalities motivated by \eqref{eq:r_ld_order} are
satisfied:
\begin{equation}
 \label{eq:scenario_of_interest}
\frac{K_{L-1}^{\mathrm{LD}}}{N_{L-1}^{\mathrm{LD}}}>\cdots
\frac{K_1^{\mathrm{LD}}}{N_1^{\mathrm{LD}}}>
\frac{K^{\mathrm{HD}}}{N^{\mathrm{HD}}}\:.
\end{equation}
The transmitted signals of the $K_l^{\mathrm{LD}}$ users within each LD
class~$l$ are formed by means of first linearly combining their data symbols
using a rectangular {\it combining} matrix $\mathbf{W}_l^{\mathrm{LD}}$ of
dimensions $N_l^{\mathrm{LD}}\times K_l^{\mathrm{LD}}$, before spreading the
resulting symbols using the columns of matrix $\mathbf{U}_l^{\mathrm{LD}}$.
Applying this class-dependent hierarchical spreading, the final
spreading code $\mathbf{c}_k$ of a LD user~$k$ can be written as
\begin{equation}
 \label{eq:spreading_code_LD}
 \mathbf{c}_k=
\left[\mathbf{U}_l^{\mathrm{LD}}\mathbf{W}_l^{\mathrm{LD}}
\right]_{j_k^{\mathrm{LD}}},\quad k\in\mathcal{K}_l^{\mathrm{LD}},
l\in\{1 \cdots L-1\}\:,
\end{equation}
where $j_k^{\mathrm{LD}}\in\{1,2,\ldots,K_l^{\mathrm{LD}}\}$ is the index of
the column of the matrix $\mathbf{U}_l^{\mathrm{LD}}\mathbf{W}_l^{\mathrm{LD}}$
assigned to user~$k\in\mathcal{K}_l^{\mathrm{LD}}$ and where
$\left[\mathbf{M}\right]_{i,j}$ designates the $i$-th element of the $j$-th
column of matrix $\mathbf{M}$.
In principle, $\mathbf{W}_l^{\mathrm{LD}}$ could be any
$N_l^{\mathrm{LD}}\times K_l^{\mathrm{LD}}$ matrix chosen such that the transmit
power constraint is respected:
\begin{equation}
 \label{eq:w_ld_cond}
 \forall j\in\mathcal{K}_l^{\mathrm{LD}},
\sum_{u=1}^{N_l^{\mathrm{LD}}}
\left|\left[\mathbf{W}_l^{\mathrm{LD}}\right]_{u,j}\right|^2=1\:.
\end{equation}
In the following, we assume that the components of $\mathbf{W}_l^{\mathrm{LD}}$
are chosen as realizations of independent and identically distributed (i.i.d.)
random variables that can take the values $+\frac{1}{\sqrt{N_l^{\mathrm{LD}}}}$
and $-\frac{1}{\sqrt{N_l^{\mathrm{LD}}}}$ with equal probabilities.
In contrast to LD users, no combining is applied to the symbols of
the users from the HD class. The spreading code $\mathbf{c}_k$ used on the
signals of a user $k\in\mathcal{K}^{\mathrm{HD}}$ can thus be written as
\begin{equation}
 \label{eq:spreading_code_HD}
 \mathbf{c}_k=\left[\mathbf{U}^{\mathrm{HD}}\right]_{j_k^{\mathrm{HD}}},\quad
k\in\mathcal{K}^{\mathrm{HD}}\:,
\end{equation}
where $\left[\mathbf{M}\right]_j$ designates the $j$-th column of matrix
$\mathbf{M}$ and where
$j_k^{\mathrm{HD}}\in\{1,2,\ldots,K^{\mathrm{HD}}\}$ is the index of the
code assigned to user~$k\in\mathcal{K}^{\mathrm{HD}}$ from among
the columns of $\mathbf{U}^{\mathrm{HD}}$.
\begin{remark}
Class dependent hierarchical spreading is distinct from
the 2-step spreading (the so-called channelization and scrambling
steps~\cite{umts}) used in third-generation (3G) systems.
Indeed, multiplication with $\mathbf{W}_l^{\mathrm{LD}}$
is intended to precondition signals from a number of users so that they can
be spread using a {\it much smaller} number of orthogonal codes. This operation
is thus completely unrelated, in its conception and in its purpose,
to both channelization and scrambling in 3G.
\end{remark}
\subsection{MOMA-OFDM}

An orthogonal frequency division multiplexing (OFDM) implementation of MOMA
(that we designate as MOMA-OFDM) consists in mapping the $N$ symbols from users'
spread signals to $N\leq N_{\mathrm{FFT}}$ consecutive subcarriers (SCs) in one
OFDM symbol, where $N_{\mathrm{FFT}}$ is the total number of SCs per OFDM
symbol. Let $\mathcal{N}$ be a subset of SCs chosen such that
$\left|\mathcal{N}\right|=N$. Let $P_k$ designate the average transmit power of
user $k\in\mathcal{K}$ . The signal transmitted by user~$k$ on subcarrier
$n\in\mathcal{N}$ is given by
\begin{equation}
 \label{eq:mc_cdma_tx_signal}
 x_{k,n}=\sqrt{P_k}\left[\mathbf{c}_k\right]_n s_k\:,
\end{equation}
where $\left[\mathbf{c}_k\right]_n$ designates the component of $\mathbf{c}_k$
mapped to subcarrier $n$ and where $s_k$ is the zero-mean
unit-variance data symbol transmitted by user~$k$.
Spreading codes $\mathbf{c}_k$ should be normalized such that
$\mathbb{E}\left[\left|x_{k,n}\right|^2\right]=\frac{P_k}{N}$.
MOMA-OFDM combines the benefits of both frequency-domain spreading, e.g.
the ability to harvest the frequency diversity of the channel and the
robustness against carrier frequency shifts, and of OFDM transmission, e.g.
robustness against timing errors.

\subsection{MOMA with Massive MIMO Base Stations}
We now turn our attention to the case where BS is equipped with a large number
$M>>1$ of antennas while each user is equipped with a single
antenna\footnote{Extension to the case of multi-antenna user terminals
is possible.}. This massive MIMO scenario,
which is expected to be prevalent in 5G networks, proves
to be particularly advantageous for MOMA, from both the performance and the
receiver complexity perspectives. We designate this implementation of MOMA as
MIMO-MOMA.
Due to the spatial multiplexing capabilities inherent to this scenario, the
number of HD users would typically be larger than the number of available HD
orthogonal codes, i.e. $K^{\mathrm{HD}}>N^{\mathrm{HD}}$.
The $N^{\mathrm{HD}}$ columns of matrix $\mathbf{U}^{\mathrm{HD}}$
should thus be shared among the $K^{\mathrm{HD}}$ HD users in such a way that
each column is reused by
$\left\lceil K^{\mathrm{HD}}/N^{\mathrm{HD}}\right\rceil$ users.

\section{MOMA Receiver}
\label{sec:M_MIMO}

Denote by $\mathbf{h}_{k,n}$ the vector of frequency-domain small-scale
fading coefficient at subcarrier $n$ between user~$k$ and the $M$ antennas of
the BS during the current OFDM symbol and assume that the components of
$\mathbf{h}_{k,n}$ are zero-mean i.i.d. random variables and that
$\forall a\in\{1,2,\ldots,M\}, \forall k\in\mathcal{K}$,
$\forall n,m\in\mathcal{N}$, $\mathbb{E}\left[\left[
\mathbf{h}_{k,n}\right]_a^*\left[\mathbf{h}_{k,m}\right]_a\right]=
c_{n-m}^h$ where $\{c_u^h\}_{u\in\mathbb{Z}}$ is a frequency domain
autocorrelation sequence.
Finally, assume that coefficients $\mathbf{h}_{k,n}$ can be estimated
at the BS by relying on uplink pilot sequences sent by the different user
terminals. The vector $\mathbf{y}_n$ of samples received at the $M$ BS antennas 
at subcarrier $n$ is given by
\begin{equation}
\label{eq:rec_samples}
 \mathbf{y}_{n}=\sum_{k\in\mathcal{K}}\sqrt{g_k P_k}
\mathbf{h}_{k,n}\left[\mathbf{c}_k\right]_n s_k+\mathbf{v}_n\:,
\end{equation}
where $\mathbf{v}_n$ is a $M\times1$ vector of i.i.d.
$\mathcal{CN}\left(0,\sigma^2\right)$ noise samples and $g_k$ is the large-scale
fading factor. The proposed receiver scheme consists in performing the following
operations.

{\bf Spatial demultiplexing} We propose the use of linear receive combining,
e.g. according to maximum-ratio (MRC) or
minimum-mean-square-error (MMSE) criteria.
Combining with coefficients $\mathbf{d}_{k,n}$ provides the samples
$r_{k,n}\stackrel{\mathrm{def}}{=}
\frac{1}{M}\mathbf{d}_{k,n}^{\mathrm{H}}\mathbf{y}_n$:
\begin{equation}
  r_{k,n}=\sqrt{g_k P_k}\left[\tilde{\mathbf{c}}_k\right]_n s_k+
\sum_{j\neq k}\sqrt{g_j P_j}
\left[\tilde{\mathbf{c}}_j\right]_n s_j+
\left[\tilde{\mathbf{v}}_k\right]_n\:,
 \label{eq:demux_samples}
\end{equation}
where we defined for any $k,j\in\mathcal{K}$ and
$\{n_1,n_2,\ldots,n_N\}=\mathcal{N}$
\begin{equation}
 \label{eq:cj_def}
 \tilde{\mathbf{c}}_j\stackrel{\mathrm{def}}{=}
\frac{1}{M}
\left[\mathbf{d}_{k,n_1}^{\mathrm{H}}\mathbf{h}_{j,n_1}
\left[\mathbf{c}_j\right]_{n_1} \cdots
\mathbf{d}_{k,n_N}^{\mathrm{H}}\mathbf{h}_{j,n_N}
\left[\mathbf{c}_j\right]_{n_N}
\right]^{\mathrm{T}}
\end{equation}
\begin{equation}
\tilde{\mathbf{v}}_k\stackrel{\mathrm{def}}{=}\frac{1}{M}
\left[\mathbf{d}_{k,n_1}^{\mathrm{H}}\mathbf{v}_{n_1}
\cdots
\mathbf{d}_{k,n_N}^{\mathrm{H}}\mathbf{v}_{n_N}
\right]^{\mathrm{T}}
\end{equation} 

In the general case where $\{h_{k,n}\}_{n\in\mathcal{N}}$ are not
fully correlated, users' channels are selective in frequency. This implies that
the effective spreading codes $\tilde{\mathbf{c}}_k$ and $\tilde{\mathbf{c}}_j$
of users $k$ and $j$ from two different classes are no longer orthogonal.

\begin{remark}[Channel Hardening]
The effect of receive combining is averaging out small-scale fading over the
array, in the sense that the variance of the effective scalar channel gain
$\frac{1}{M}\mathbf{d}_{k,n}^{\mathrm{H}}\mathbf{h}_{k,n}$ decreases with
$M$. This effect is known as channel hardening and is a consequence of
the law of large numbers~\cite{mimo_myths}. the frequency response
$\frac{1}{M}\mathbf{d}_{k,n}^{\mathrm{H}}\mathbf{h}_{k,n}$ of the effective
channel is thus asymptotically flat with $M$ and the above-mentioned loss of
orthogonality vanishes.
\end{remark}  

{\bf HD Users Detection:} In order to obtain HD data rates
that are as high as possible, multiuser detection should be used. However,
we limit the use of multiuser
detection to cases where it is effectively needed. We know from
the literature~\cite{mimo_cdma_tse} that, under the assumption of uncorrelated
antenna channel coefficients, the effective HD spreading gain is
$M N^{\mathrm{HD}}$. Therefore, we propose the use of MMSE with successive
interference cancellation (SIC) {\it only} in the
case where $K^{\mathrm{HD}}\sim M N^{\mathrm{HD}}$, as opposed to
$K^{\mathrm{HD}}\ll M N^{\mathrm{HD}}$ for which single-user detection should be
sufficient.
Define $\mathbf{A}\stackrel{\mathrm{def}}{=}\mathrm{diag}\left\{
\sqrt{P_k g_k}\right\}_{k\in\mathcal{K}^{\mathrm{HD}}}$ and the matrix
$\tilde{\mathbf{C}}$ of effective codes as
$\tilde{\mathbf{C}}\stackrel{\mathrm{def}}{=}
\left[\tilde{\mathbf{c}}_{k_1}
\cdots\tilde{\mathbf{c}}_{k_{K^{\mathrm{HD}}}}\right]$,
where $\{k_1,k_2,\ldots,k_{K^{\mathrm{HD}}}\}=\mathcal{K}^{\mathrm{HD}}$.
Now assume that the columns of $\mathbf{A}$ and $\tilde{\mathbf{C}}$ are
arranged in the descending order with respect to the values $\sqrt{P_k g_k}$
and define $\mathbf{T}\stackrel{\mathrm{def}}{=}\tilde{\mathbf{C}}\mathbf{A}$.
The MMSE-SIC receiver starts by recovering the symbol of user $k_1$ for which
the value $\sqrt{P_{k_1}g_{k_1}}$ is the largest by computing
$r_{k_i}^{\mathrm{HD}}\stackrel{\mathrm{def}}{=}
\boldsymbol{\delta}_{k_i}^{\mathrm{H}}\mathbf{r}_{k_i}$ where
$\boldsymbol{\delta}_k\stackrel{\mathrm{def}}{=}
\left(\mathbf{T}\mathbf{T}^{\mathrm{H}}+
\left(\sigma^2+\sum_{j\in\mathcal{K}^{\mathrm{LD}}}g_j P_j\right)/M
\mathbf{I}\right)^{-1}[\mathbf{T}]_k$ and 
\begin{equation}
 \label{eq:despreading_samples_mc_moma}
 \begin{multlined}
   (K^{\mathrm{HD}}\sim M N^{\mathrm{HD}})\quad   
r_{k_i}^{\mathrm{HD}}=
\sqrt{g_{k_i}P_{k_i}}\boldsymbol{\delta}_{k_i}^{\mathrm{H}}
\tilde{\mathbf{c}}_{k_i}s_{k_i}+
\boldsymbol{\delta}_{k_i}^{\mathrm{H}}\tilde{\mathbf{v}}_{k_i}+\\
\sum_{j=i+1}^{K^{\mathrm{HD}}}\sqrt{g_{k_j}P_{k_j}}
\boldsymbol{\delta}_{k_i}^{\mathrm{H}}\tilde{\mathbf{c}}_{k_j} s_{k_j}+
\sum_{j\in\mathcal{K}^{\mathrm{LD}}}\sqrt{g_j P_j}
\boldsymbol{\delta}_{k_i}^{\mathrm{H}}\tilde{\mathbf{c}}_j s_j\:,
 \end{multlined}
\end{equation}
for $k_i = k_1$ and $\mathbf{r}_{k_1}\stackrel{\mathrm{def}}{=}
\left[r_{k_1,n_1} r_{k_1,n_2}\cdots r_{k_1,n_{N}}\right]^{\mathrm{T}}$.
Once $s_{k_1}$ is decoded correctly, the contribution of~$k_1$ can be removed
from $\{r_{k_2,n}\}_{n\in\mathcal{N}}$ in order to detect the data
symbol of user~$k_2$. The $N$-long vector of signal samples after this
cancellation is denoted as $\mathbf{r}_{k_2}$.
This SIC procedure continues till the detection of all the
HD data symbols. In the case where $K^{\mathrm{HD}}<<M\times N^{\mathrm{HD}}$,
MMSE-SIC is dropped and single-user (SU) detection is instead used by computing
$r_{k}^{\mathrm{HD,SU}}\stackrel{\mathrm{def}}{=}
\mathbf{c}_k^{\mathrm{H}}\mathbf{r}_k$:
\begin{equation}
 \label{eq:despreading_samples_mc_moma_su}
 \begin{multlined}
(K^{\mathrm{HD}}\ll M N^{\mathrm{HD}})\quad
r_{k}^{\mathrm{HD,SU}}=
\sqrt{g_{k}P_{k}}\mathbf{c}_k^{\mathrm{H}}
\tilde{\mathbf{c}}_{k}s_{k}+\\
\sum_{j\in\mathcal{K}\setminus\{k\}}\sqrt{g_j P_j}
\mathbf{c}_k^{\mathrm{H}}\tilde{\mathbf{c}}_j s_j+
\mathbf{c}_k^{\mathrm{H}}\tilde{\mathbf{v}}_{k}\:.
 \end{multlined}
\end{equation}

{\bf LD User Detection:} We propose the use of single-user detection for LD
users, so that for any $k\in\mathcal{K}^{\mathrm{LD}}$, detection
is based on the decision sample $r_k^{\mathrm{LD}}\stackrel{\mathrm{def}}{=}
\mathbf{c}_{k}^{\mathrm{H}}\mathbf{r}_k$ given by
\begin{equation}
 \label{eq:despreading_samples_mc_moma_ld}
r_k^{\mathrm{LD}}=
\sqrt{g_k P_k}\mathbf{c}_{k}^{\mathrm{H}}\tilde{\mathbf{c}}_k s_k+
\sum_{j\neq k}
\sqrt{g_j P_j}\mathbf{c}_{k}^{\mathrm{H}}
\tilde{\mathbf{c}}_j s_j+\mathbf{c}_k^{\mathrm{H}}\tilde{\mathbf{v}}_k\:.
\end{equation}

\subsection{Detection Signal to Noise Plus Interference Ratio}

The signal to noise plus interference ratio (SINR) of any LD
user~$k\in\mathcal{K}_l^{\mathrm{LD}}$ can be derived
from (\ref{eq:despreading_samples_mc_moma_ld}) and is given by
\begin{equation}
 \label{eq:sinr_ld_sic}
 \begin{multlined}
  \mathrm{SINR}_k^{\mathrm{LD}}\stackrel{\mathrm{def}}{=}
\frac{g_k P_k\left|\mathbf{c}_{k}^{\mathrm{H}}
\tilde{\mathbf{c}}_k\right|^2}{
\sum_{j\in\mathcal{K}\setminus\{k\}}^{K}
g_j P_j\left|\mathbf{c}_{k}^{\mathrm{H}}
\tilde{\mathbf{c}}_j\right|^2+\sigma_{\mathrm{LD}}^2}\:,
 \end{multlined}
\end{equation}
where $\sigma_{\mathrm{LD}}^2\stackrel{\mathrm{def}}{=}
\frac{1}{M^2}\sum_{n\in\mathcal{N}}\left|\left[\mathbf{c}_k\right]_n\right|^2
\mathbf{d}_{k,n}^{\mathrm{H}}
\mathbf{d}_{k,n}\sigma^2$.
The SINR for any $k\in\mathcal{K}^{\mathrm{HD}}$ when MMSE-SIC is used
follows from \eqref{eq:despreading_samples_mc_moma} as
\begin{equation}
 \label{eq:sinr_hd_sic}
 \begin{multlined}
  \mathrm{SINR}_{k_i}^{\mathrm{HD}}\stackrel{\mathrm{def}}{=}\\
\frac{g_{k_i} P_{k_i}
\left|\boldsymbol{\delta}_{k_i}^{\mathrm{H}}
\tilde{\mathbf{c}}_{k_i}\right|^2}{\sum_{j=i+1}^{K^{\mathrm{HD}}}g_{k_j}
P_{k_j}
\left|\boldsymbol{\delta}_{k_i}^{\mathrm{H}}
\tilde{\mathbf{c}}_{k_j}\right|^2+
\sum_{j\in\mathcal{K}^{\mathrm{LD}}}g_j P_j
\left|\boldsymbol{\delta}_{k_i}^{\mathrm{H}}
\tilde{\mathbf{c}}_j\right|^2+\sigma_{\mathrm{HD}}^2}
 \end{multlined}
\end{equation}
where $\sigma_{\mathrm{HD}}^2\stackrel{\mathrm{def}}{=}
\frac{1}{M^2}\sum_{n\in\mathcal{N}}\left|\left[
\boldsymbol{\delta}_k\right]_n\right|^2
\mathbf{d}_{k,n}^{\mathrm{H}}\mathbf{d}_{k,n}\sigma^2$.
While in the case of single-user detection
\eqref{eq:despreading_samples_mc_moma_su} gives rise to
\begin{equation}
 \label{eq:sinr_su}
  \mathrm{SINR}_{k}^{\mathrm{HD,SU}}
\stackrel{\mathrm{def}}{=}
\frac{g_k P_k\left|\mathbf{c}_k^{\mathrm{H}}\tilde{\mathbf{c}}_k\right|^2}{
\sum_{j\in\mathcal{K}\setminus\{k\}}g_j P_j
\left|\mathbf{c}_k^{\mathrm{H}}\tilde{\mathbf{c}}_j\right|^2+
\sigma_{\mathrm{HD}}^2}\:,
\end{equation}
Finally, define the {\it instantaneous} bits/s/Hz
capacities\footnote{log denotes the base-2 logarithm.} 
$R_k^{\mathrm{LD}}\stackrel{\mathrm{def}}{=}
\log\left(1+\mathrm{SINR}_k^{\mathrm{LD}}\right)$ and
$R_k^{\mathrm{HD,SU}}\stackrel{\mathrm{def}}{=}
\log\left(1+\mathrm{SINR}_k^{\mathrm{HD,SU}}\right)$.

\subsection{Asymptotic Analysis}

The following theorem states that both inter-class and HD intra-class
interference become asymptotically negligible as $M$ increases even if only
single-user detection is employed.
\begin{theo}
 \label{theo:cdma_massive_mimo_equivalence}
Assume that the components of each $\mathbf{W}_l^{\mathrm{LD}}$ are realizations
of i.i.d. zero-mean random variables that satisfy the condition in
\eqref{eq:w_ld_cond}. Also assume that the empirical distribution of the
large-scale fading coefficients $\{g_k\}_{k\in\mathcal{K}}$ converges as
$K\to\infty$ to the distribution of a random variable with mean $\mathbb{E}[g]$.
Finally, $\forall k\in\mathcal{K}^{\mathrm{LD}}, P_k=P^{\mathrm{LD}}$.
If $\frac{K^{\mathrm{LD}}}{M}\to_{M\to\infty}\alpha$ while
$K^{\mathrm{HD}}=\mathcal{O}_M(1)$ and $N\ll N_{\mathrm{FFT}}$, then we have as
$M\to\infty$
\begin{align}
 &(k\in\mathcal{K}^{\mathrm{HD}})
&\sum_{j\in\mathcal{K}\setminus\{k\}}g_j P_j
\left|\mathbf{c}_k^{\mathrm{H}}\tilde{\mathbf{c}}_j\right|^2
\stackrel{p}{\rightarrow}0\:,\label{eq:theo1}\\
&(k\in\mathcal{K}_l^{\mathrm{LD}})
&\sum_{j\in\mathcal{K}\setminus\{k\}}g_j P_j
\left|\mathbf{c}_k^{\mathrm{H}}\tilde{\mathbf{c}}_j\right|^2-
\frac{c K_l^{\mathrm{LD}}}{N_l^{\mathrm{LD}}M}
\stackrel{p}{\rightarrow}0\:,\label{eq:theo2}\\
&(k\in\mathcal{K})
&\mathbf{c}_k^{\mathrm{H}}\tilde{\mathbf{c}}_k
\stackrel{a.s.}{\rightarrow}1\:,\label{eq:theo3}
\end{align}
where $c\stackrel{\mathrm{def}}{=}P^{\mathrm{LD}}\mathbb{E}[g]$ and where
$\stackrel{p}{\rightarrow}$ and $\stackrel{a.s.}{\rightarrow}$ stand for
convergence in probability and for almost sure convergence of random variables,
respectively.
\end{theo}
\begin{proof}
Assume that $L=2$ so that we can drop from now on the use of the index $l$.
This assumption is only made for the sake of ease of presentation as the
following proof arguments apply in the general case of $L\geq2$.

To show that \eqref{eq:theo1} holds, we write its left-hand side as
\begin{align}
 \label{eq:hd_mui}
 &\sum_{j\in\mathcal{K}\setminus\{k\}}g_j P_j
\left|\mathbf{c}_k^{\mathrm{H}}\tilde{\mathbf{c}}_j\right|^2=\nonumber\\
&\sum_{j\in\mathcal{K}^{\mathrm{HD}}\setminus\{k\}}g_j P_j
\left|\sum_{n\in\mathcal{N}}
\left[\mathbf{c}_k\right]_n
\frac{1}{M}\mathbf{h}_{k,n}^{\mathrm{H}}\mathbf{h}_{j,n}
\left[\mathbf{c}_j\right]_n\right|^2+\nonumber\\
&\sum_{j\in\mathcal{K}^{\mathrm{LD}}}g_j P_j
\left|\sum_{n\in\mathcal{N}}
\left[\mathbf{c}_k\right]_n
\frac{1}{M}\mathbf{h}_{k,n}^{\mathrm{H}}\mathbf{h}_{j,n}^{\mathrm{H}}
\sum_{u=1}^{N^{\mathrm{LD}}}\left[\mathbf{U}^{\mathrm{LD}}\right]_{n,u}
\left[\mathbf{W}^{\mathrm{LD}}\right]_{u,j}
\right|^2
\end{align}
The first term in \eqref{eq:hd_mui} converges almost surely (a.s.) to zero as
$M\to\infty$ due to applying the law of large numbers to the sum
$\frac{1}{M}\mathbf{h}_{k,n}^{\mathrm{H}}\mathbf{h}_{j,n}$.
As for the second term, it can be rewritten by referring to
\eqref{eq:spreading_code_LD} as
\begin{align}
 \label{eq:2t}
  &\sum_{j\in\mathcal{K}^{\mathrm{LD}}}g_j P_j
\left|\sum_{n\in\mathcal{N}}\left[\mathbf{c}_k\right]_n
\frac{1}{M}\mathbf{h}_{k,n}^{\mathrm{H}}\mathbf{h}_{j,n}
\sum_{u=1}^{N^{\mathrm{LD}}}\left[\mathbf{U}^{\mathrm{LD}}\right]_{n,u}
\left[\mathbf{W}^{\mathrm{LD}}\right]_{u,j}\right|^2\nonumber\\
&=P^{\mathrm{LD}}\sum_{u,v=1}^{N^{\mathrm{LD}}}\sum_{n,m\in\mathcal{N}}
\left[\mathbf{c}_k\right]_n\left[\mathbf{c}_k\right]_m^*
\left[\mathbf{U}^{\mathrm{LD}}\right]_{n,u}
\left[\mathbf{U}^{\mathrm{LD}}\right]_{m,v}\times\nonumber\\
&\frac{1}{M^2}\sum_{j\in\mathcal{K}^{\mathrm{LD}}}g_j
\mathbf{h}_{k,n}^{\mathrm{H}}\mathbf{h}_{j,n}
\mathbf{h}_{j,m}^{\mathrm{H}}\mathbf{h}_{k,m}
\left[\mathbf{W}^{\mathrm{LD}}\right]_{u,j}
\left[\mathbf{W}^{\mathrm{LD}}\right]_{v,j}^*\:.
\end{align}
Defining $\forall n\in\mathcal{N}$, $\xi_{k,j,n}\stackrel{\mathrm{def}}{=}
\frac{1}{\sqrt{M}}\mathbf{h}_{k,n}^{\mathrm{H}}\mathbf{h}_{j,n}$ we can write 
\begin{equation}
 \label{eq:sum_l_to_p_inter}
 \begin{multlined}
\frac{1}{M^2}\sum_{j\in\mathcal{K}^{\mathrm{LD}}}g_j
\mathbf{h}_{k,n}^{\mathrm{H}}\mathbf{h}_{j,n}
\mathbf{h}_{j,m}^{\mathrm{H}}\mathbf{h}_{k,m}
\left[\mathbf{W}^{\mathrm{LD}}\right]_{u,j}
\left[\mathbf{W}^{\mathrm{LD}}\right]_{v,j}^*=\\
\frac{K^{\mathrm{LD}}}{M}\frac{1}{K^{\mathrm{LD}}}
\sum_{j\in\mathcal{K}^{\mathrm{LD}}}g_j \xi_{k,j,n}\xi_{k,j,m}^*
\left[\mathbf{W}^{\mathrm{LD}}\right]_{u,j}
\left[\mathbf{W}^{\mathrm{LD}}\right]_{v,j}^*\:.
 \end{multlined}
\end{equation}
Now, thanks to the assumption that the empirical distribution of
$\{g_k\}_{k\in\mathcal{K}}$ converges as $K\to\infty$
to the distribution of a random variable with mean $\mathbb{E}[g]$ and that
the components of $\mathbf{W}^{\mathrm{LD}}$ are realizations of i.i.d.
zero-mean random variables, the arguments of the proof of Proposition 3.3
from~\cite{effective_bandwidth} can be applied to show, after tedious but
straightforward steps, that for each value of
$(n,m,u,v)\in\mathcal{N}^2\times
\{1,2,\ldots,N^{\mathrm{LD}}\}^2$
\begin{equation}
 \label{eq:sum_l_to_p}
 \begin{multlined}
 \frac{1}{K^{\mathrm{LD}}}
\sum_{j\in\mathcal{K}^{\mathrm{LD}}}g_j \xi_{k,j,n}\xi_{k,j,m}^*
\left[\mathbf{W}^{\mathrm{LD}}\right]_{u,j}
\left[\mathbf{W}^{\mathrm{LD}}\right]_{v,j}^*\stackrel{p}{\rightarrow}\\
\mathbb{E}\left[g\right]\mathbb{E}\left[\xi_{k,j,n}\xi_{k,j,m}^*\right]
\mathbb{E}\left[\left[\mathbf{W}^{\mathrm{LD}}\right]_{u,j}
\left[\mathbf{W}^{\mathrm{LD}}\right]_{v,j}^*\right]=\\
\frac{1}{N^{\mathrm{LD}}}\mathbb{E}\left[g\right]
\left|c_{n-m}^h\right|^2\delta_{u,v}\:,
 \end{multlined}
\end{equation}
where $\{c_u^h\}_{u\in\mathbb{Z}}$ is the frequency domain autocorrelation
sequence of users' channels and where $\delta_{u,v}=1$ if $u=v$ and
$\delta_{u,v}=0$ otherwise.
Note that $\forall n,m\in\mathcal{N}, \left|c_{n-m}^h\right|^2\approx1$ thanks
to the assumption that $N\ll N_{\mathrm{FFT}}$.
Plugging $\left|c_{n-m}^h\right|^2=1$
and \eqref{eq:sum_l_to_p} into \eqref{eq:2t}, we get
\begin{align}
  &\sum_{j\in\mathcal{K}^{\mathrm{LD}}}g_j P_j
\left|\sum_{n\in\mathcal{N}}\left[\mathbf{c}_k\right]_n
\frac{1}{M}\mathbf{h}_{k,n}^{\mathrm{H}}\mathbf{h}_{j,n}
\sum_{u=1}^{N^{\mathrm{LD}}}\left[\mathbf{U}^{\mathrm{LD}}\right]_{n,u}
\left[\mathbf{W}^{\mathrm{LD}}\right]_{u,j}\right|^2\nonumber\\
&\stackrel{p}{\rightarrow}
\frac{\alpha P^{\mathrm{LD}}\mathbb{E}\left[g\right]}{N^{\mathrm{LD}}}
\sum_{u=1}^{N^{\mathrm{LD}}}\sum_{n,m\in\mathcal{N}}
\left[\mathbf{c}_k\right]_n\left[\mathbf{c}_k\right]_m^*
\left[\mathbf{U}^{\mathrm{LD}}\right]_{n,u}
\left[\mathbf{U}^{\mathrm{LD}}\right]_{m,u}\nonumber\\
&=\frac{\alpha P^{\mathrm{LD}}\mathbb{E}\left[g\right]}{N^{\mathrm{LD}}}
\sum_{u=1}^{N^{\mathrm{LD}}}\left|\mathbf{c}_k^{\mathrm{H}}
\left[\mathbf{U}^{\mathrm{LD}}\right]_u\right|^2\nonumber\\
&=0\:,
\end{align}
where the last equality follows from the fact that the HD spreading code
$\mathbf{c}_k$ is orthogonal by construction to
$\left[\mathbf{U}^{\mathrm{LD}}\right]_u$ for any
$u\in\{1,2,\ldots,N^{\mathrm{LD}}\}$.
Using similar arguments, one can show that \eqref{eq:theo2} holds true.
Finally, \eqref{eq:theo3} holds due to \eqref{eq:spreading_code_LD} and
\eqref{eq:spreading_code_HD} and to the law of large numbers
applied to the sum $\frac{1}{M}\mathbf{h}_{k,n}^{\mathrm{H}}\mathbf{h}_{k,n}$.
This completes the proof of Theorem~\ref{theo:cdma_massive_mimo_equivalence}.
\end{proof}
Applying Theorem~\ref{theo:cdma_massive_mimo_equivalence} along with the
continuous-mapping theorem to~\eqref{eq:sinr_su} reveals that
$R_k^{\mathrm{HD,SU}}$ is increasing with $M$ in an unbounded manner.
As for LD users, the following result can be used to properly tune
$P^{\mathrm{LD}}$, $K_l^{\mathrm{LD}}$ and $N_l^{\mathrm{LD}}$ so that
$\forall l\in\{1\cdots L-1\}$ the target $r_l^{\mathrm{LD}}$ is
achieved.
\begin{corollary}
 Under the assumptions made in Theorem~\ref{theo:cdma_massive_mimo_equivalence},
$R_k^{\mathrm{LD}}-\log\left(1+\mathrm{SINR}_k^{\mathrm{LD},\infty}\right)
\stackrel{p}{\rightarrow}0$, where
\begin{equation}
 \label{eq:sinr_infty}
  \forall k\in\mathcal{K}_l^{\mathrm{LD}},\quad
\mathrm{SINR}_k^{\mathrm{LD},\infty}\stackrel{\mathrm{def}}{=}
\frac{g_k P^{\mathrm{LD}}}{\frac{c K_l^{\mathrm{LD}}}{N_l^{\mathrm{LD}}M}+
\frac{\sigma^2}{M}}\:.
\end{equation}
\end{corollary}

\section{Numerical Results}
\label{sec:simus}

Simulations results were obtained assuming users' distances to the
BS are randomly chosen from the interval $\left[25,100\right]$~m and that the
associated pathloss coefficients $g_k$ are computed using the COST-231 Hata
model~\cite{cost_hata} with a carrier frequency $f_0=900$ MHz. Users' transmit
power is equal to 23 dBm while the noise power spectral density is equal to
$N_0=-174$ dBm/Hz. Two channel models are considered, namely the Extended
Pedestrian A (EPA) and the Extended Vehicular A (EVA)
models~\cite{3gpp_etu_epa}. Channels generated using the EPA model have smaller
delay spreads, and hence frequency responses that are less selective, than
the EVA model.
For the underlying OFDM system we assume a total number $N_{\mathrm{FFT}}=1,024$
of subcarriers out of which, as in LTE, only $N_{\mathrm{SC}}=600$ SCs are
used for data transmission.
Furthermore, we consider a 2-class MOMA-OFDM, i.e. $L=2$, that is implemented on
the basis of $N=32$ subcarriers using a $32\times32$ Walsh-Hadamard matrix,
i.e. 19 instances of MOMA-OFDM are needed to cover the $N_{\mathrm{SC}}=600$
available SCs. Out of the $N$ orthogonal codes, $N^{\mathrm{HD}}=\frac{7}{8N}$
are reserved for HD users and $N^{\mathrm{LD}}=\frac{1}{8N}$ for LD users.
This partition was chosen for the purposes of fair comparison with LTE-M.
Indeed, in the latter system 72 SCs are reserved for M2M transmissions,
corresponding to $\frac{72}{600}\approx\frac{1}{8}$ of the available
$N_{\mathrm{SC}}=600$ SCs.
Finally, $K^{\mathrm{HD}}=8N^{\mathrm{HD}}$ corresponding to a spatial
multiplexing gain of 8. 
All the following results have been obtained by averaging over 100 realizations
of users' random positions in the cell.

Fig.~\ref{fig:kLD_rLD_mux} shows the average number of LD users that can be
simultaneously served using MIMO-MOMA as function of the LD data rate
requirement $r^{\mathrm{LD}}$ when compared to both LoRa and a
narrow-band cellular IoT system implemented using LTE-M parameters. From the
figure we can notice the significant advantage of using MOMA as opposed to
narrow-band access schemes in terms of IoT network capacity. Indeed, the
resources reserved in the latter systems for M2M communications turn out to be
under-used when compared to MOMA. Also note that the performance gap between a
narrow-band cellular IoT system and MOMA is the largest on the range of low to
moderate LD target data rates. As for the range of high target data rates
(on which LoRa slightly outperforms the 2-class implementation of MOMA),
{\bf introducing a third class for high-rate M2M transmissions} can in
principle cancel this performance gap.
\begin{figure}[h]
 \centering
 \includegraphics[width=1\hsize]{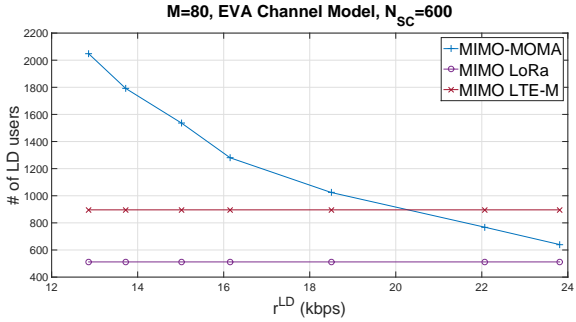}
 \caption{Number of served LD users within one OFDM symbol vs. the
LD data rate requirement. $M=80$ with spatial multiplexing.}
 \label{fig:kLD_rLD_mux}
\end{figure}
It is worth mentioning that the maximum number of simultaneous M2M transmissions
in LoRa was computed assuming both a spatial multiplexing gain of 8 (for fair
comparison with our MIMO-MOMA setting) and the availability of 16
125kHz-channels, each of which can support up to 7 concurrent transmissions
thanks to the multiplexing capabilities of chirp spread
spectrum~\cite{lora_phy}\footnote{In LoRa, only concurrent
transmissions using different spreading factors, and hence having
different data rates on the range 0.3$\sim$50 kbps, are orthogonal.
This explains why the maximum number of simultaneous transmissions in LoRa is
plotted as a constant function with respect to the target data rate.}. 

Finally, Fig.~\ref{fig:rHD_mux} shows that MIMO-MOMA achieves HD data rates that
are very close to the maximum value achievable with perfect orthogonality.
For instance, the HD data rate achieved by MOMA on EPA channels (resp. on EVA
channels) with single-user detection stays withing 99\% (resp. within 86\%) of
the perfect-orthogonality upper bound.
Note that this performance is achieved by MOMA while serving a
number of concurrent LD transmissions as large as 4 times the number that can be
supported with narrow-band cellular IoT solutions.
\begin{figure}[h]
 \centering
 \includegraphics[width=1\hsize]{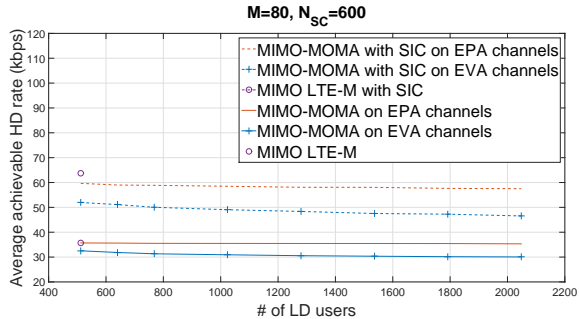}
 \caption{Achievable rate of HD users vs. the number of served LD users
within one OFDM symbol. $M=80$ with spatial multiplexing.}
 \label{fig:rHD_mux}
\end{figure}

\section{Conclusion}
In this article, we presented a novel multiple access scheme (MOMA) for
next-generation cellular networks that is fully compatible with massive MIMO.
This scheme is based on assigning, in a flexible and dynamic manner, different
resources and different degrees of resource overloading to different
classes of users, each representing a different data rate requirement and/or
a different service type. Moreover, transmissions from different classes in
MOMA are quasi-orthogonal. This way, the use of non-orthogonal access for the
lower-rate classes would only slightly affect the broadband users, dropping the
need for wasteful guard bands, for subband-based filtering or for complex
receiver schemes.


\end{document}